\documentclass[11pt]{article}

\usepackage{microtype}
\usepackage{cite}
\usepackage{algorithm}
\usepackage{float}
\usepackage[noend]{algpseudocode}
\makeatletter
\def\BState{\State\hskip-\ALG@thistlm}
\makeatother

\usepackage[margin=1in]{geometry}
\usepackage{graphicx}
\usepackage{multicol}
\usepackage{url}  
\usepackage{accents}  
\usepackage{wrapfig}

\usepackage{amssymb}                    
\usepackage{amsfonts}                   
\usepackage{amsthm}                     

\newtheorem{theorem}{\textbf{Theorem}}[section]
\newtheorem{lemma}{\textbf{Lemma}}[section]

\newtheorem{definition}[theorem]{Definition}
\newtheorem{property}[theorem]{Property}

\newcommand{\polylog}{\texttt{polylog}}

\newcommand{\Real}{\mathbb{R}}

\begin{document}

\title{Convex Hull for Probabilistic Points}

\author{
{\small
\begin{tabular}{ccc}
  {\large  F. Betul Atalay} & {\large Sorelle A. Friedler} & {\large Dianna Xu} \\
  \\
  Dept. of Computer Engineering & Dept. of Computer Science & Dept. of Computer Science \\
  TOBB University of Economics & Haverford College & Bryn Mawr College \\
  and Technology & 370 W. Lancaster Avenue & 101 N. Merion Avenue \\
 Sogutozu, Ankara, Turkey & Haverford, PA 19041, USA & Bryn Mawr, PA 19010, USA \\
  fatalay@etu.edu.tr & sorelle@cs.haverford.edu & dxu@cs.brynmawr.edu\\
\end{tabular}
}}
\date{}

\maketitle

\begin{abstract}
We analyze the correctness of an $O(n \log n)$ time divide-and-conquer algorithm for the convex hull problem when each input point is a location determined by a normal distribution.  
We show that the algorithm finds the convex hull of such probabilistic points to precision within some expected correctness determined by a user-given confidence value $\phi$.  In order to precisely explain how correct the resulting structure is, we introduce a new certificate error model for calculating and understanding approximate geometric error based on the fundamental properties of a geometric structure.  We show that this new error model implies correctness under a robust statistical error model, in which each point lies within the hull with probability at least $\phi$, for the convex hull problem. 
\end{abstract}

\section{Introduction}

The \emph{Convex Hull Problem} is the problem of determining a minimum convex bounding polygon that covers $n$ points in the Euclidean plane. 
\begin{figure}[h]
\begin{center}
 \includegraphics[width=0.20\textwidth]{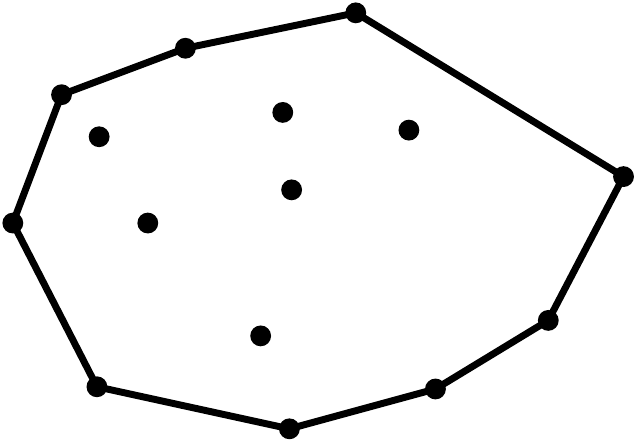}
\caption{A point set and its convex hull}
\label{fig:chull}
\vspace{-0.1in}
\end{center}
\end{figure}

This is a classic problem in computational geometry, with well known solutions including Graham's scan and divide-and-conquer (both take $O(n \log n)$ time) \cite{cgbook, ORourke1998book}.  The convex hull is a fundamental primitive for many graphics problems, such as calculation of basic shape representations (e.g., bounding boxes) \cite{imageProcBook2014} and collision detection \cite{collisionDetection2001}.  In application domains, point locations are often the result of a machine learning algorithm that outputs a probability distribution for each point's location (for a survey, see \cite{Hightower01Location}).  For example, in augmented reality, the markerless tracking problem that aims to track the position and orientation of a camera in a scene without using markers may take a hybrid approach that relies on both computer vision techniques and probabilistic GPS location information of the type generated by such machine learning algorithms \cite{augrealitySurvey2008}.  In this paper, we are interested in examining what happens when the expected values of such probabilistic points are given as input to the divide-and-conquer convex hull algorithm, with the goal of guaranteeing approximate correctness of the resulting convex hull without requiring extensive modification to existing algorithms.

We will show that the divide-and-conquer convex hull algorithm still produces an approximately correct convex hull even when its input point locations aren't known exactly.  This will require some modifications to the algorithm as well as an introduction of a new error model in order to define what we mean by an approximately correct convex hull.  We will build this new approximate notion on boolean functions that certify geometric properties necessary to a correct calculation of the convex hull.  These functions are borrowed from the study of kinetic data structures \cite{Basch99MobileData}, and so some of this work will find application to other problems studied within that boolean certification framework (as well as allowing future work to extend these results to hold on moving points).  A careful analysis will show how potential errors in these certifications propagate to the overall structure being calculated.  The convex hull will be approximate in the sense that only a given percent of the points will be expected to lie within it.  This matches the desires of some applications - for example, when determining the home range of an animal from (noisy) location observations, the goal is to compute a boundary containing some percentage of such observations \cite{animalHomeRanges2010}.

\subsection{Related Work}

Approximate correctness of a geometric structure has been considered under a number of different models, including interpretations where the structure is considered to be fully correct some percentage of time or where it is considered to be partially correct every time the algorithm is run.  We are most interested in this second interpretation, within which partial correctness has been considered under the absolute error model \cite{Fonseca10Approximate}, the relative error model \cite{Arya00Approximate}, and the robust error model \cite{Rousseeuw05Robust}.
Within the \emph{absolute error model} a structure is considered to be correct up to some given fixed error bound $\varepsilon$ that is constant for any set of points \cite{Fonseca10Approximate}.   
Under the \emph{relative error model} a structure is considered to be correct up to some percentage based on the geometric structure \cite{Arya00Approximate}.
The \emph{robust error model} is a per-point error model under which a structure is correct based on the percentage of points which are correct \cite{Rousseeuw05Robust}.  We will compare the error model we introduce to the robust error model. 

While classical computational geometry assumes exact knowledge of point location, goals of relaxing such assumptions have spurred several recent papers.
Loeffler and Kreveld \cite{Loeffler10Imprecise} have considered approximate convex hulls under an imprecise point setting, where exact point location is unknown within a region but guaranteed not to be outside of it.  They consider the convex hull under multiple variants of the relative error model and achieve running times that range from $O(n \log n)$ to $O(n^{13})$.
When considering approximate nearest neighbor searching, a model where points are described as probability distributions over their possible locations has also been considered \cite{Agarwal13ANN}.  This latter model of point location, commonly used in application domains, is the same as the one we use here (and is described in more detail in Section \ref{sec:prob_points}).  The convex hull problem has been considered within the discrete version of this point location model (where the distributions are discrete) by Agarwal \textit{et al.} \cite{Agarwal14ConvexHulls}.  
Their results give a running time of $O(m \log^3 m)$, where $m$ is the number of possible point locations in their discrete distributions.  The robust error model in Argarwal \textit{et al} implies the one we compute.  While we solve a weaker version of the problem, we improve the running time to $O(n \log n)$, where $n$ is the number of points.  Additionally, ours is the first solution to hold on continuous distributions.  

We use an $O(n \log n)$ divide-and-conquer algorithm to compute the convex
hull on a set of probabilistic points under normal distributions.
Our solution is approximately correct under a robust error model with
the correctness taken in expectation  over all possible point locations, 
so that each point has at least $\phi$ probability of being in the hull, 
for a parameter $\phi$. 
Ours is the first solution to hold for
probabilistic points with a continuous location probability
distribution.

We achieve these results not by introducing a new algorithm,
but by introducing a new error model and associated analysis of the
standard divide-and-conquer algorithm for calculating the convex hull
via its upper envelope in the dual space \cite{ORourke1998book}.  
We introduce a \emph{certificate error model} in which a
structure is considered $\phi$-correct if each Boolean certificate
used to calculate the structure is correct with probability at least
$\phi$.  We will show that approximate correctness under the
certificate error model implies approximate correctness under the
robust error model for the convex hull.

\subsection{Contributions}

The rest of this paper shows the following results:
\begin{enumerate}
\item We introduce a certificate error model guaranteeing that each certificate is correct with probability $\phi$, and a proof that this new error model implies the robust error model for the convex hull problem.  (See Sections \ref{certificate_error_model} and \ref{convex_hull_proof}.)
\item We adapt an $O(n \log n)$ algorithm to compute the convex hull for probabilistic points.  We show that this algorithm is approximately correct in expectation over all possible point locations, under a robust error model guaranteeing that each point is within the hull with probability at least $\phi$.  (See Section \ref{convex_hull}.)
\end{enumerate}

\section{Preliminaries}
\subsection{Probabilistic Points}
\label{sec:prob_points}

We define a \emph{probabilistic point} $p_j = ({N_j}, v_j)$ where ${N_j}$ is a normal probability distribution over its possible locations and $v_j \in \Real^d$ is an \emph{expected value} for the point $p_j$ given distribution ${N_j}$.  We are given a set $P$ of $n$ probabilistic points. 
$D_j = \{ x \in \Real^d | N_j^{pdf}(x) > 0 \}$ is the positive region of the probability density function $N_j^{pdf}:\Real^d \rightarrow \{ y \in \Real | y \geq 0 \}$.  We assume that the region $D_j$ is bounded.  Let $\beta_j(\phi, D_j) \subset \Real^d$ be the \emph{boundary region} of point $p_j$ defined as the minimum-area convex set such that $p_j$ is within the region with probability $\phi$, i.e.,
$\int\limits_{x \in \beta_j(\phi, D_j)} N_j^{pdf}(x) dx = \phi$.  $\phi \in [0,1]$ is a user-given confidence value, and $\Phi = 100 \cdot \phi$ is $\phi$ in percent form.  We assume that $\beta_j(\phi, D_j)$ can be calculated in $O(1)$ time.  For example, Figure \ref{fig:google_maps} shows $\beta_j(\phi, D_j)$ as the truncated Gaussian.

For the remainder of this paper we will refer to these probabilistic points $p_j = (D_j, v_j)$ simply as \emph{points}. 
Within a machine learning context, these points would be generated by a \emph{model} $M(j, E) \rightarrow p_j$ that, when given a point identity $j$ and environmental data $E$, would return the probabilistic point $p_j$.  More details about such models can be found in a survey of location models \cite{Hightower01Location}.  We will assume that this model is good enough that the point locations can generally be distinguished from each other, i.e., that for $p_i, p_j \in P$, drawn from the distribution created by the model,
\[ Pr[ \beta_i(\phi, D_i) \cap \beta_j(\phi, D_j) = \emptyset] \geq \phi  \mbox{ .}\]

\begin{figure}
\begin{center}
 \includegraphics[width=0.24\textwidth]{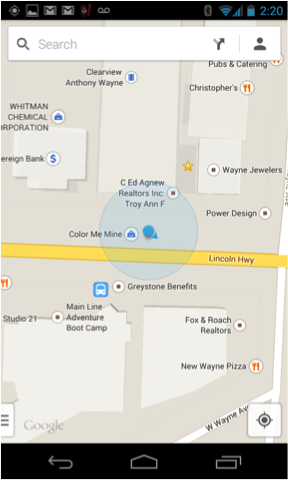}
\caption{A Google Maps screenshot showing a probabilistic point $p_j$ under a normal distribution $D_j$ where the central blue dot is $v_j$ and the lighter blue circle is its associated $\beta_j(\phi, D_j)$.}
\label{fig:google_maps}
\vspace{-0.1in}
\end{center}
\end{figure}

\subsection{Certificates}
\label{intro:certificates}
\looseness-1 Given a set of probabilistic points, we develop a framework that approximately maintains a geometric structure $G$ up to some expected correctness. We define a set of certificates $C$ that guarantee local geometric relationships crucial to the correctness of the entire structure.  For example, a single certificate might guarantee that three points are oriented in a counter-clockwise relationship.  $C$ can be considered a \emph{proof} of the correctness of $G$. These certificates are the same as those  maintained in classic kinetic data structure (KDS) settings \cite{Basch99MobileData} (we will extend them later).  The set $C$ consists of pairs containing a Boolean function $c$ which operates on a set of points $P_i \subset P$ and the set of points $P_i$ on which that function evaluates to \texttt{True}.  Such a pair $(c, P_i)$ is called a \emph{certificate}.  Within a single set $C$, there can be multiple types of such functions $c$, certifying different geometric properties.  For notational ease we will abuse notation below and refer to all such functions as $c$.  A set of certificates $C$ must satisfy the following local geometric properties  as given in \cite{Basch99MobileData}.

\begin{property}[Locality]
\label{def:locality}
For all points $p_j \in P$, $| \{ P_i ~|~ p_j \in P_i \mbox{ and } (c,P_i) \in C \} |$ is $O(\polylog(n))$ or $O(n^\epsilon)$ for arbitrarily small values of $\epsilon$.
\end{property} 

\begin{property}[Compactness]
\label{def:compactness}
$|C|$ is $O(n ~\polylog(n))$ or $O(n^{1 + \epsilon})$ for arbitrarily small $\epsilon$.
\end{property}

\begin{property}[Exclusivity]
\label{def:exclusivity}
$|P_i| \leq k$ for all $(c,P_i) \in C$, $P_i \subset P$,  and small constant $k$.
\end{property}

Locality and compactness are both required within the KDS framework and exclusivity is also generally assumed \cite{Basch99MobileData}.  Thus, we can draw on a large body of existing work defining certificates for a wide variety of problems.  (See \cite{Guibas04Kinetic}).  Notably, these certificates certify the steps of certain locally constrained algorithms and incrementally constructed problem solutions.  Divide and conquer algorithms often make good candidates for such problem certification mechanisms; Each decision in the merge process constitutes a certificate.

We add to the KDS understanding of certificates to take into account the probabilistic nature of the points.  
\begin{definition}[$\phi$-correct certificate]
\label{def:certificate}
A certificate $(c, P_i)$ for which
\[ Pr[c(P_i) = \texttt{True}] \geq \phi \]
with the probability taken over the distribution of possible point locations for points $p_j = (N_j, v_j)$ for $p_j \in P_i$. 
\end{definition}

\noindent For example, a simple certificate $(above_\phi, P_i)$ with $P_i = \{p_1, p_2 \}$ certifies that $p_1$ is above $p_2$ with probability at least $\phi$.  (See Section \ref{certificate_error_model} for a more extensive example of a problem using such certificates.)  It will be useful to note that $v_{j} \in \beta_j(\phi, D_j)$ for all $p_j \in P_i$ since $N_j$ is a normal distribution.  If all certificates are $\phi$-correct for $\phi = 1$, then the geometric structure $G$ has been correctly calculated.  The main motivation of this paper is to consider the correctness for values of $\phi < 1$.

Given knowledge of $\beta(P_i) = \{ \beta_j(\phi, D_j) | p_j \in P_i \}$, we now determine the correctness of certificate $(c, P_i)$.  $\phi^k$-correctness can be achieved by creating certificates $(c',P_i)$ with new function $c'$ such that  $c'(P_i) = \texttt{True}$ if and only if for all possible point locations 
$P_i = \{_{j=1}^k ~ p_j \in  \beta_j(\phi, D_j) ~|~ \beta_j(\phi, D_j) \in \beta(P_i) \}$
we have $c'(P_i) = \texttt{True}$.  This can be easily improved to $\phi$-correctness by determining $\beta_j(\phi^{1/k}, D_j)$ instead.  However, this is a conservative lower bound on the correctness of the certificate.  In the example certificate $above_\phi(P_i)$, we would be guaranteeing that $p_1$ is above $p_2$ and that $\beta_1(\phi, D_1)$ does not intersect $\beta_2(\phi, D_2)$.  Instead, we could calculate directly the probability that $p_1$ is above $p_2$ and set $above'_\phi(P_i) = \texttt{True}$ as long as that probability is at least $\phi$.  This guarantees that $above'_\phi(P_i)$ is $\phi$-correct.

\section{Convex Hull Algorithm}
\label{convex_hull}
Recall that the convex hull is defined as the smallest convex region
containing a set of points.  In order to determine certificates that
guarantee a solution to this problem, we turn to the KDS definition of
convex hull certificates \cite{Basch99MobileData} that we will review in this section.  The KDS solution for this problem makes use
of a divide and conquer algorithm to find the convex hull via finding the 
upper and lower envelopes in the
dual setting, where a point $(a, b)$ is represented by the line $y =
ax + b$ \footnote{Note that standard notation dualizes $(a, b)$ to $y=ax-b$, 
however, KDS \cite{Basch99MobileData} uses $ax+b$, which we follow to avoid confusion when discussing certificates.}
\cite{Basch99MobileData, ORourke1998book}.  

Given a set of $n$ lines $L= \{l_1, l_2, \ldots, l_n\}$ where $l_i$ is of the form
$y=a_ix+b_i$, if we think of these
lines as defining $n$ halfplanes, $y \ge a_ix+b_i$, each lying \emph{above}
one of the lines, then the \emph{upper envelope} of $L$ is the boundary of the intersection of these half planes (see Figure \ref{fig:envelope}). The \emph{lower envelope} is defined symmetrically.
\begin{figure}[H]
\begin{center}
\includegraphics[width=0.24\textwidth]{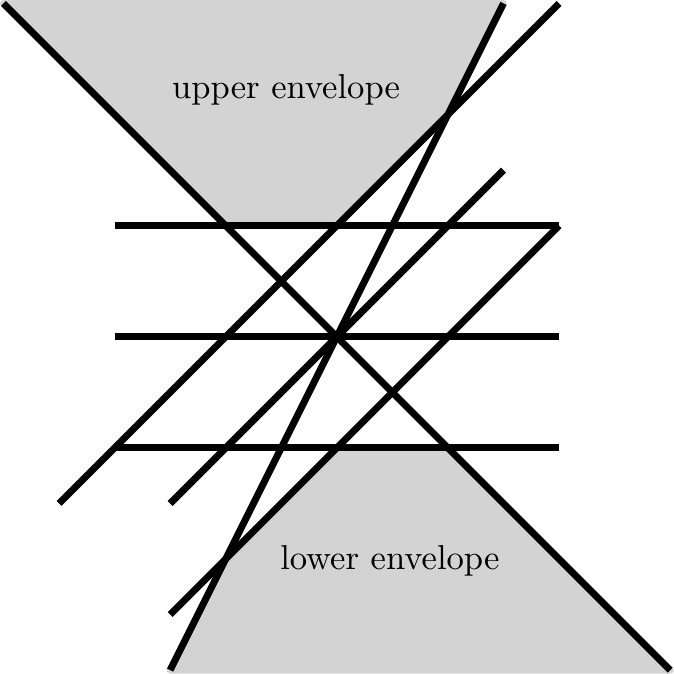}
\caption{Upper and lower envelopes}
\label{fig:envelope}
\vspace{-0.2in}
\end{center}
\end{figure}

Classic computational geometry has a well-established equivalency of the
convex hull of points and the upper/lower envelopes of a collection of lines
under the point-line duality transformation \cite{Basch99MobileData, ORourke1998book}, in that the clockwise order of the points along the upper (lower) convex hull of a set of points $P$ is equal to the left-to-right order of the sequence of the lines on the upper (lower) envelope of the dual $P^*$ (see Figure \ref{fig:chull-env}). 

\begin{figure}[H]
\begin{center}
\includegraphics[width=0.4\textwidth]{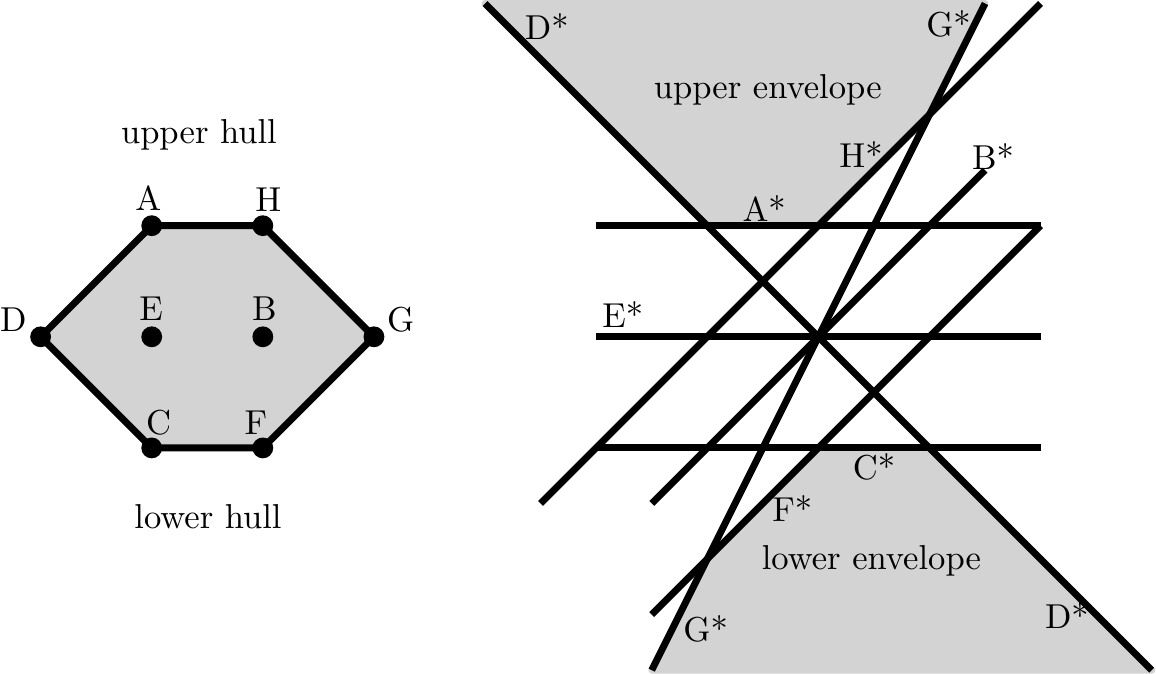}
\caption{Equivalence of convex hulls and envelopes}
\label{fig:chull-env}
\end{center}
\end{figure}

\begin{figure*}[!ht]  
\begin{center}
\includegraphics[width=\textwidth]{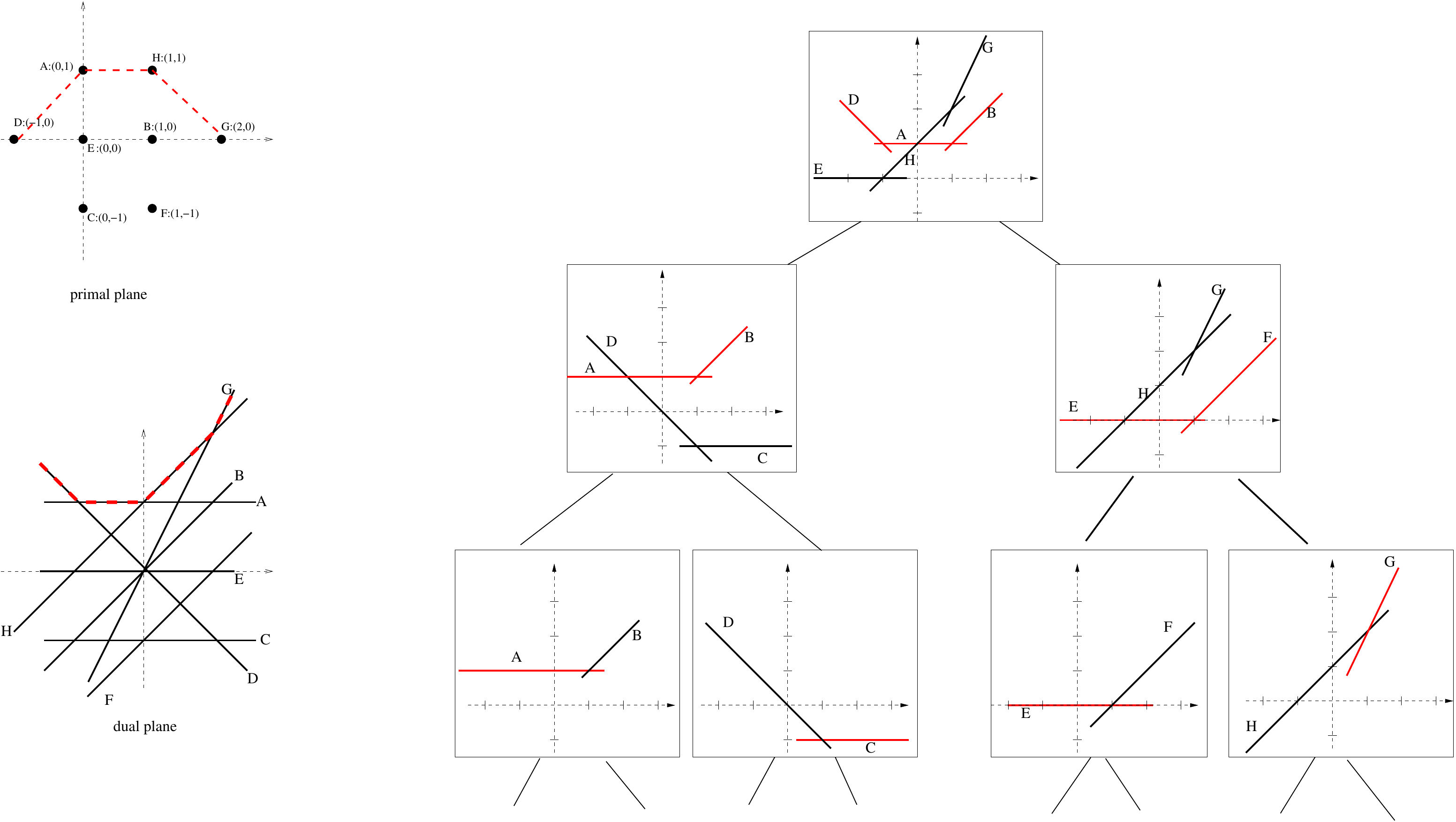}
\end{center}
\caption{Top-left: points in primal plane. Bottom-left: dual plane,
  where a point $(a, b)$ is represented by the line $y = ax + b$. Right: the merge tree corresponding to the upper envelope computation in dual space. Leaf nodes are single lines and are omitted in the figure. The certificates proving the top-most merge  are as follows: (i) a chain intersection certificate guaranteeing that $EH <_x AB$ and $EH <_y A$ and $AB <_y H$, that is, EH is to the left of AB, EH is below line A, AB is below line H (ii) a diverging certificate guaranteeing that $B {\leq}_s G$ and $AB <_y G$, that is, B's slope is less than or equal to G's slope and vertex AB is below line G, (iii) and, another diverging certificate guaranteeing that $E {\leq}_s D$ and $EH <_y D$. The certificates are explained in more detail in the paper introducing them \cite{Basch99MobileData}.}

\label{fig:mergetree}
\end{figure*}

The KDS algorithm computes the convex hull in the dual because it allows easier management of certificates.
Since our analysis will rely on certificates, we will make use of
this algorithm as well. Please refer to 
Figure \ref{fig:mergetree} as we describe it briefly below. 

We focus only on the upper envelope of dual lines: the lower envelope
computation is symmetric.  To find the upper envelope of a set of
lines, the lines are partitioned into two subsets of roughly equal
size, their upper envelopes are computed recursively, and then the two
resulting upper envelopes---let them be called red and black--- are
merged.  The merge step is performed by sweeping a vertical line
through all of the vertices of the red and black upper envelopes, from
left to right. As the line sweeps, the most recently encountered red
and black vertices are maintained along with the information whether
the red or the black chain is above. As the sweep encounters the next
red (black) vertex, the algorithm determines if it is above or below
the corresponding black (red) edge. If it is above, the current vertex
is added to the merged envelope.  If the above/below ordering of the
envelopes has changed, that means the red and black envelopes have
crossed, and the intersection point is also added to the merged upper
envelope.  This algorithm takes time $O(n \log n)$
\cite{ORourke1998book}.  We use the same algorithm and return the same
points as the hull, so our algorithm also takes time $O(n \log n)$.
However, our certificates are $\phi$-correct, so we will need to more
carefully consider the correctness of the resulting hull (Section
\ref{convex_hull_proof}).  First, we will describe these certificates
in more detail.

As originally presented in \cite{Basch99MobileData}, the comparisons
done during the merge step lead to the following certificates:
(i) \emph{x-certificates} ($<_x$) are used to certify the x-ordering of
the red and black vertices, (ii) \emph{y-certificates} ($<_y$) are
used to certify the y-ordering of a vertex with respect to an edge of
the opposite color. (iii) \emph{slope-certificates} (${\leq}_s$) involve
comparisons between line slopes. Slope-certificates are not required
in the above sweep algorithm, however they are needed in order to make
the KDS local---avoiding linearly many y-certificates per edge.

\begin{figure*}[!ht]  
\begin{minipage}{0.3\textwidth}
\begin{center}
\includegraphics[height=1.2cm]{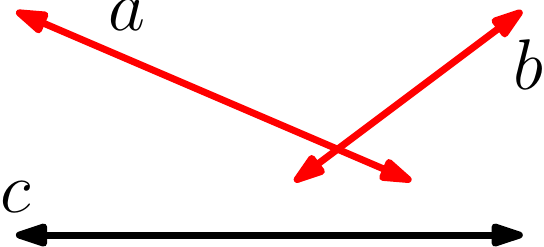}
\end{center}
\end{minipage}
\begin{minipage}{0.3\textwidth}
\begin{center}
\includegraphics[height=1.2cm]{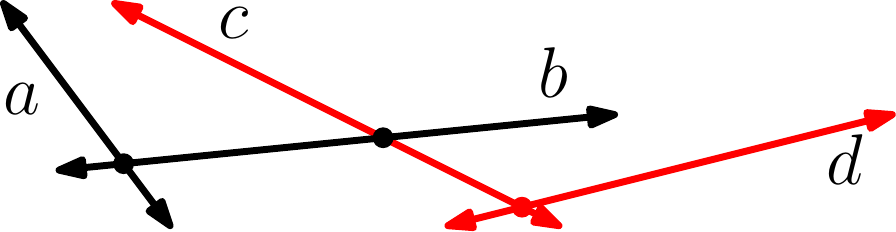}
\end{center}
\end{minipage}
\begin{minipage}{0.3\textwidth}
\begin{center}
\includegraphics[height=1.2cm]{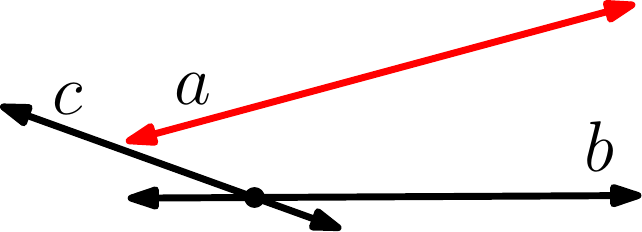}
\end{center}
\end{minipage}
\caption{Convex hull certificates \cite{Basch99MobileData} in the dual setting showing the two chains involved in red and black.  Left: Tangent certificates guarantee that line $c$ is below the vertex $ab$ and that the slope of line $c$ is greater than the slope of $a$ and less than the slope of $b$.  Center: Intersection certificates guarantee that the vertex $ab$ is to the left of vertex $cd$ and below line $c$ and that vertex $cd$ is to the right of vertex $ab$ and below line $b$.  Right: Diverging certificates guarantee that $b$'s slope is less than or equal to the slope of line $a$ and that the vertex $cb$ is below line $a$.  
}
\label{fig:certificate_types}
\end{figure*}

In a \emph{merge tree}, we keep track of all levels of the divide and conquer algorithm's merge step and certify the properties that determine each choice in the merge of two recursively determined upper envelope chains. 
Recall from \cite{Basch99MobileData} and \cite{Atallah85Dynamic} that the leaf nodes of the merge tree are single lines and the root node contains all line segments determining the resulting upper envelope.  
 The certificates, as presented in \cite{Basch99MobileData}, come in dependent groups. That is, the invalidation of any one certificate invalidates all other certificates in that group, and the existence of any one certificate implies the existence of all other certificates in the group.  
These certificates are \emph{tangent certificates},  
\emph{chain intersection certificates}, 
and \emph{diverging certificates} 
(see Figure \ref{fig:certificate_types}).  Diverging and tangent
certificates represent the two ways in which an edge will not be
involved in the resulting upper envelope while intersection
certificates represent a merge point of the two chains.  For more
precise details defining these certificates see
\cite{Basch99MobileData}.  One assumption implicit in the presentation
of these certificates and in the statement of the original divide and
conquer algorithm is that the points are in general
position, i.e., no two lines in the dual setting are
parallel.  
We make the same assumption. See Figure \ref{fig:mergetree} for an
example involving eight points (shown both in primal space and the dual
space), and the merge tree corresponding
to the  upper envelope computation.

\paragraph*{Convex Hull Algorithm} 
Using these certificates, we now have the following algorithm.  Create the certificates based on the expected values of the points.  We will show in Lemma \ref{lem:expected_to_phi} this gives a set of $\phi$-correct certificates.  Find the convex hull using these certificates and the $O(n \log n)$ divide and conquer algorithm \cite{Basch99MobileData}.  Then find the boundary of the convex hull of those points.  Specifically, given points $H$ determined by the certification process as the hull, we report $\mathcal{CH}(\{\beta_i(\phi, D_i) ~|~ p_i \in H\})$ where $\mathcal{CH}$ takes the convex hull of the convex bounding regions of the points on the hull (e.g., using the algorithm of \cite{devillers95incremental} which takes time $O(n \log n)$).  The resulting algorithm takes time $O(n \log n)$.

 \section{Certificate Error Model}
\label{certificate_error_model}
In order to reason about the correctness of this convex hull algorithm when expected values are used instead of precise points, we introduce a new error model that evaluates the correctness of a geometric structure based on the correctness of its component certificates.

\begin{definition}[$\phi$-correct within the certificate error model]
  Given a set of certificates $C$ guaranteeing a
  geometric structure, for all $(c,P_i) \in C$, $(c,P_i)$ is
  $\phi$-correct.
\end{definition}

When we originally construct the set of certificates for a problem we will do so based on the set of expected values of the points, $V(P_i) = \{ v_{j} | p_j \in P_i \}$.  The question then becomes what is the relationship between certificates involving the expected values and those involving the set of points at their full distribution of locations.  Certificates that are 1-correct based on the expected values are easily achieved by making the direct comparisons based on the known expected values, the question is how these relate to the certificates involving the points. 
Recall that $v_i \in \beta_i(\phi, D_i)$ for normal distributions.

\begin{lemma}
\label{lem:expected_to_phi} Given that 
\[ Pr[ \beta_j(\phi, D_j) \cap \beta_k(\phi, D_k) = \emptyset] \geq \phi \mbox{ for all } p_j, p_k \in P_i \]
where $P_i \subset P$, then
\[ \big( Pr[c(V(P_i)) = \texttt{True}] = 1 \big) \Longrightarrow \big( Pr[ c(P_i) = \texttt{True}] \geq \phi \big) \mbox{ .} \]
\end{lemma}

\begin{proof}
  Consider the cases when $c(P_i) = \texttt{False}$: either
  $\beta_j(\phi, D_j) \cap \beta_k(\phi, D_k) \not= \emptyset$
  or at least one of the true locations $p'_j$ associated with $p_j \in P_i$ is outside of
  $\beta_j(\phi, D_j)$.  Both of these cases occur with probability at most $1 - \phi$, so the certificate probability guarantee of $\phi$ has been verified.
\end{proof}

With this lemma, we have the following basic procedure for translating any existing set of certificates into a $\phi$-correct set:
\begin{enumerate}
\item Create a set of $1$-correct certificates based on the expected values of the points.  These certificates are a set of $\phi$-correct certificates involving points, by Lemma \ref{lem:expected_to_phi}.
\item Solve the problem using these certificates as usual.
\item \label{step:worst_case} Return a worst case result based on the boundaries of the points in the solution.
\end{enumerate}
Step \ref{step:worst_case} will be explained in more detail for each specific problem solution.

Here, we begin with an example problem that demonstrates the value and structure of the certificate error model without the complexity of the convex hull certificates that will be described in Section \ref{convex_hull}.  

\subsection{1D Maximum Problem} 
The \emph{1D Maximum Problem} determines
the maximum point among a set of $n$ one-dimensional points.  
We will assume that these points are in general position. 
We will use the same certificates as those in the KDS solution to this 
problem, which relies on a max
heap that maintains the full ordering of the points.  The $n-1$
certificates guarantee the parent-child relationships in the
heap \cite{Basch99MobileData}.

Following the general procedure outlined above, parent-child certificates are created based on their expected values.  We will say that $above_\phi(p, q)$ certifies that $Pr[p > q] \geq \phi$.  The original certificates created are $above_1(v_p, v_q)$.  By Lemma \ref{lem:expected_to_phi}, this gives us certificates $above_\phi(p, q)$.
Since $v_p \in \beta_p(\phi, D_p)$ and $v_q \in \beta_q(\phi, D_q)$, we also have that
$ above_{\phi}(p,q) \Longrightarrow \big( above_\phi(p, v_q) \land above_\phi(v_p, q) \big) \mbox{.}$
We find the point $m$ that is the maximum based on these certificates and report the value $\texttt{max} = max \{ p \in \beta_m(\phi, D_m) \}$ as the result of the 1D maximum problem.

\begin{figure}[h]
\begin{center}
 \centering
\includegraphics[width=0.3\textwidth]{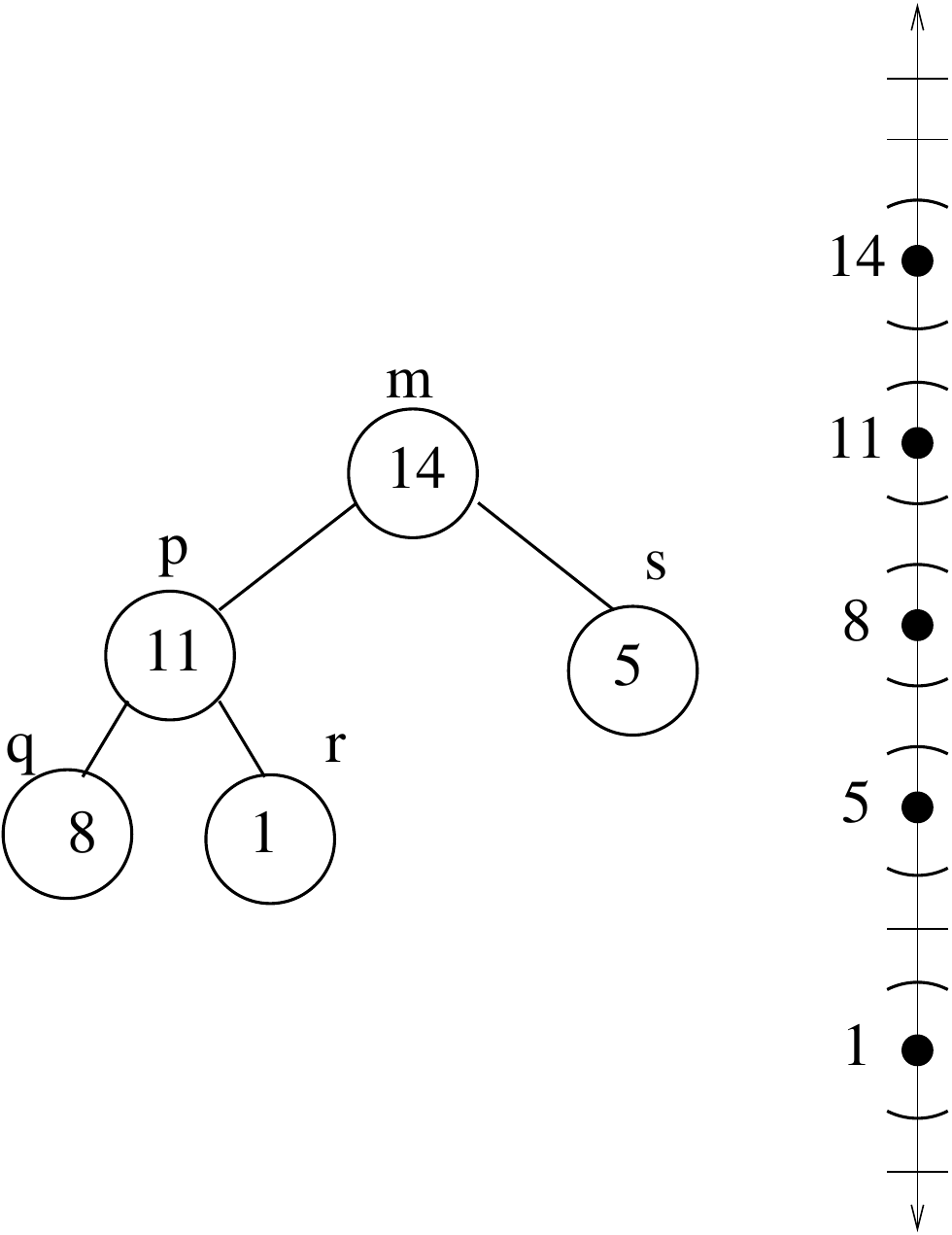}
\caption{Max-heap for the 1D Maximum problem.}
\label{fig:max1D}
\end{center}
\end{figure}

For example, consider the max heap shown in Figure \ref{fig:max1D}
maintaining probabilistic points $m, p, q, r$ and $s$.  There are four
certificates: $above_\phi(m,p)$, $above_\phi(p,q)$, $above_\phi(p,r)$ and $above_\phi(m,s)$
associated with this heap. Recall that $above_\phi(p,q)$ certifies that $p > q$ with probability at least $\phi$.  The
expected values $v_p$ at which these points are believed to be are shown within each
node.  In the case where the boundary extends 1 unit in each direction centered at $v_m$, the maximum returned is $15 = max \{ \beta_m(\phi, D_\phi) \}$.

However, the true value of a point may differ from what is
believed, which may cause a certificate to be incorrect. For example,
if point $q$ has a true value of 16, it makes the certificate
$above_\phi(p,q)$ false. Moreover, as in this case, an incorrect
certificate may mean that $q$'s real value (16) is higher than the
maximum given by the heap.  Our goal is to show that this bad case does not happen too often.  We will use the assumption that the  points' boundary regions may not intersect with high probability.

Specifically, we will show that $\phi$-correctness under the certificate error model for the 1D maximum problem implies correctness under a previously studied error model, the robust error model.  Under the \emph{robust error model}, more commonly known as a robust statistical estimator, a structure is robust to outliers up to some breakdown point \cite{Rousseeuw05Robust}.  Defining the error model more specifically is problem-dependent.
\begin{definition}[1D Maximum Problem: $\phi$-correct within the robust error model]
For a point set $P$ with $n$ points, the returned maximum point $\texttt{max}$ is such that
\[ | \{ p \in P ~|~ p \leq \texttt{max} \} | \geq \phi \cdot n \mbox{ .}\]
\end{definition}

Since we are working with probabilistic points, we are interested in an \emph{expected} $\phi$-correct robust error model.

\begin{definition}[1D Maximum Problem: expected $\phi$-correct within the robust error model]
For a point set $P$ with $n$ points, the returned maximum point $\texttt{max}$ is such that
\[ \mathbb{E} \left[ ~| \{ p \in P ~|~ p \leq \texttt{max} \} | ~\right] \geq \phi \cdot n \mbox{ .}\]
\end{definition}

\noindent If, for all points $p \in P$, $Pr[p \leq \texttt{max}] \geq \phi$ then by linearity of expectations, expected $\phi$-correctness is implied, with the expectation taken over the point location distributions.  We will thus proceed with the expected version of the definition.

\begin{theorem}[1D Maximum Problem: Certificate error model implies robust error model]
\label{thm:1dmax}
Given a max heap with certificate set $C$, points $P$, and returned maximum point $\texttt{max}$,
if for all $(c,P_i) \in C$, $(c, P_i)$ is $\phi$-correct
then
$\texttt{max}$ is expected $\phi$-correct within the robust error model.
\end{theorem}

\begin{proof}
Since the heap is $\phi$-correct under the certificate error model, we expect $100-\Phi$ percent of the certificates to be incorrect.  There are $n - 1$ 
certificates in the heap.  We will examine how many points can be greater than the maximum for each incorrect certificate.  If a certificate $above_\phi(p,q)$ is incorrect then $p \not\in \beta_p$ or $q \not\in \beta_q$.  We will associate each point outside of its boundary with the certificate it participates in as a child.  Recall from Definition \ref{def:certificate} that the location of $q'$ being drastically different from $v_q$ does not invalidate the child certificates of $q$ - those were still made with respect to $q = (v_q, D_q)$.   
  Thus, each incorrect certificate (associated with some point that is outside of its boundary) will cause at most its child point to be above the maximum.  The only remaining point to consider is the maximum point, which doesn't participate in any certificate as a child.  By the definition of $\texttt{max} = max \{ p \in \beta_m(\phi, D_m) \}$, the true value
  of $m$ is greater than $\texttt{max}$ with probability at most $1 - \phi$.
      
  So $n$ points will be above the maximum reported value each with probability at most $1 - \phi$.  Thus, we expect that the reported
  maximum point will be within the top $100 - \Phi$ percent of the
  points and so the result is expected to be $\phi$-correct under the
  robust error model.
\end{proof}

\section{Convex Hull Error Model Correctness}
\label{convex_hull_proof}
We will consider correctness of the hull under the robust error model which states that the convex hull is $\phi$-correct if at least $\Phi$ percent of the points are contained within the hull.  

\begin{definition}[Convex Hull: expected $\phi$-correct within the robust error model]
Given a set of $n$ points $P$ and its set of points $H$ on the convex hull, where $\mathcal{H}$ is the associated hull region,
\[ \mathbb{E} \left[ | \{ p \in P ~|~ p \in \mathcal{H} \} | \right] \geq \phi \cdot n \mbox { .} \]
\label{def:robust_convex}
\end{definition}

\noindent This is a restatement of an idea Tukey referred to as ``peeling" in which some number or percentage of outlying points are iteratively deleted and the convex hull is computed on the remaining points \cite{Huber72Robust}.  A conservative interpretation of this definition states the goal as computing the minimum convex area containing $\Phi$ percent of the points.  Given the probabilistic nature of our points, we will instead compute the minimum convex area containing the entirety of the boundary (the truncated distribution of point locations) for the points on the hull based on their expected values.  We will show that this convex hull satisfies the weaker Definition \ref{def:robust_convex} above.

The certificate error model remains the same: a convex hull is $\phi$-correct 
if each certificate is $\phi$-correct.

\newcommand{\thmConvexHull}{
Given a set of certificates $C$ certifying a convex hull solution for points $P$, with returned set of probabilistic points $H$ that make up the hull, where $\mathcal{H}$ is the associated hull region, 
\[ \forall_{(c,P_i) \in C}~, (c,P_i) \mbox{ is } \phi-correct
\longrightarrow\]
\[ \mathbb{E} \left[ | \{ p \in P ~|~ p \in \mathcal{H} \} | \right] \geq \phi \cdot n \mbox{ .} \]
}

\begin{theorem}[Convex Hull: Certificate error model implies robust error model]
\label{thm:convex_hull}
\thmConvexHull
\end{theorem}

In order to reason about the certificate correctness of the convex
hull, we will first need to understand the properties of these
certificates in more
detail.  
We will say that a point $p$ has been \emph{excluded} from the convex
hull by some certificate when believing that certificate's false
assertion causes $p$ to be outside of the resulting convex hull.  See
Figure \ref{fig:excluded point} for an example where point $D$ is
excluded from the convex hull due to an incorrect tangent
certificate. Note that if a certificate incorrectly causes $p$ to be
\emph{on} the convex hull, $p$ has not been excluded.

\begin{figure}[h]  
\begin{center}
\includegraphics[width=3in]{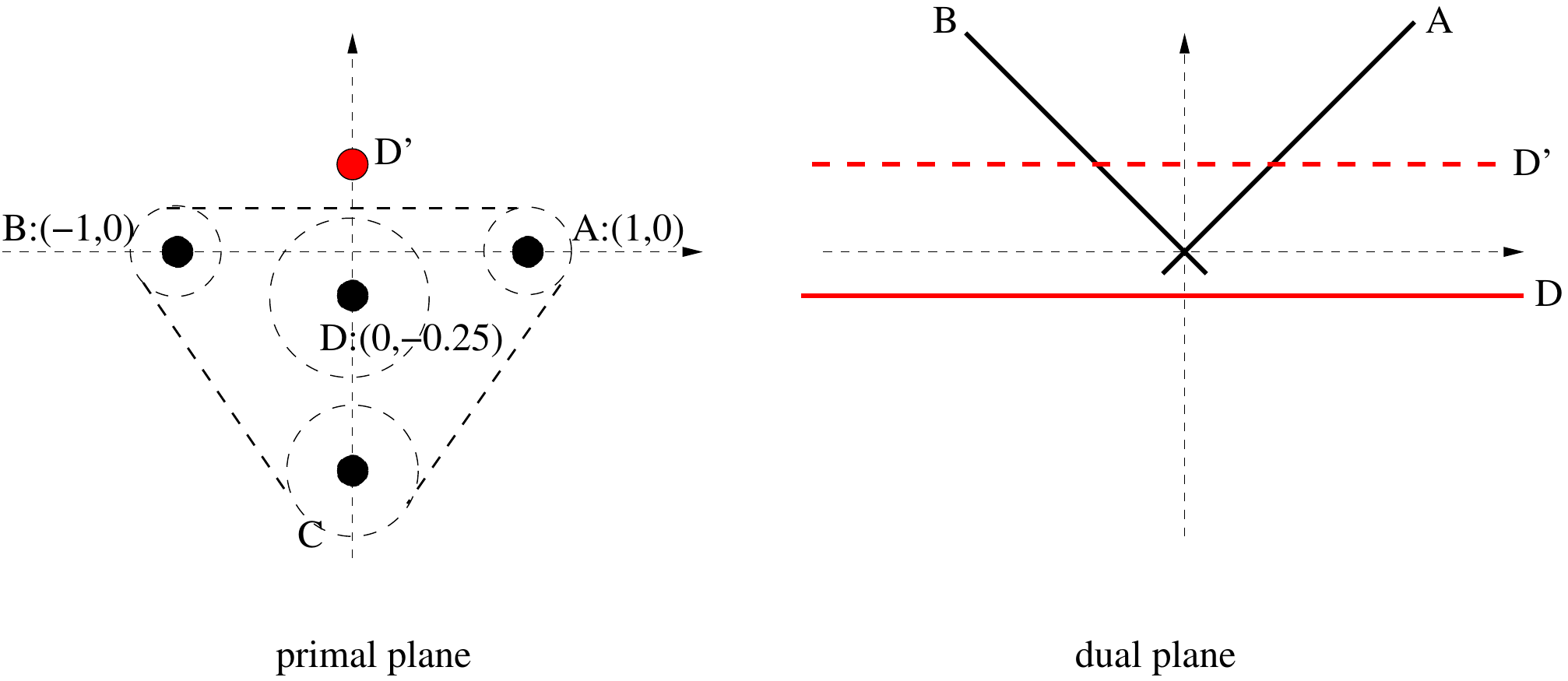}
\end{center}
\caption{Left: probabilistic points in primal plane. Right: Consider the merge of the black envelope consisting of lines A and B, and the red envelope consisting of D only. This merge is certified by a tangent certificate guaranteeing that line D is below vertex AB ($D <_y AB$), the slope of line D is greater than the slope of B ($B <_s D$), and less than the slope of A ($D <_s A$). If point D's real position is D', notice that the line corresponding to D' in dual space is above point AB, and this makes the tangent certificate incorrect. Point D is excluded from the convex hull by this incorrect tangent certificate.}

\label{fig:excluded point}
\end{figure}
Recall the tree of certificates that show the choices in the divide and conquer convex hull algorithm, which we called the merge tree (an example was shown in Figure \ref{fig:mergetree}).  
Here, we consider the levels of the merge tree in which a single point (a line in the dual setting) participates.

\begin{lemma}
\label{lem:highest_level}
Each point $p \in P$ with associated dual line $\ell$ can only be excluded from the convex hull by an incorrect certificate in the highest level $L$ of the merge tree in which $\ell$ participates.
\end{lemma}

\begin{proof}
  \looseness-1 Suppose there is an incorrect certificate that involves $\ell$ in some level $L'$ below level $L$.  Despite this incorrect certificate, $\ell$ advanced to level $L' + 1 \leq L$, so $\ell$ was found to be on the upper envelope for all lines in its subtree at level $L'$.  Points reported as on the convex hull (lines in the upper envelope in the dual setting) have not been excluded from the hull.  So as long as $\ell$ advances to some level above $L'$, incorrect certificates in level $L'$ can not exclude $\ell$.  So $\ell$ can only be excluded from the convex hull by an incorrect certificate that keeps it from being reported on the convex hull.  This level is, by construction of the merge tree, the highest level at which $\ell$ participates in the merge tree.
\end{proof}

Next, we examine the properties of certificates from a single line's point of view.

\begin{lemma}
\label{lem:one_cert}
Each point $p \in P$ with associated dual line $\ell$ can be excluded from the convex hull by at most one certificate when considering a single level $L$ of the merge tree.
\end{lemma}

\begin{proof} 
First, we consider which lines have the potential to be excluded from the hull for each type of certificate.  Referring to the line labels of Figure \ref{fig:certificate_types}, note that only $c$ may be excluded by an incorrect tangent certificate, $a$ or $d$ may be excluded due to an incorrect intersection certificate, and only $b$ may be excluded by an incorrect diverging certificate, since these are the lines not believed to be on the hull.  Note also that some of the lines in these certificates ($a$ and $d$ in the intersection certificates and $c$ in the diverging certificates) only contribute a vertex to the geometric relationship being guaranteed between the lines.  We will say that such lines \emph{participate as a vertex} in the certificate, while the other lines will be said to \emph{participate as a line}.

\begin{figure}[htbp]
   \centering
   \includegraphics[width=0.3\textwidth]{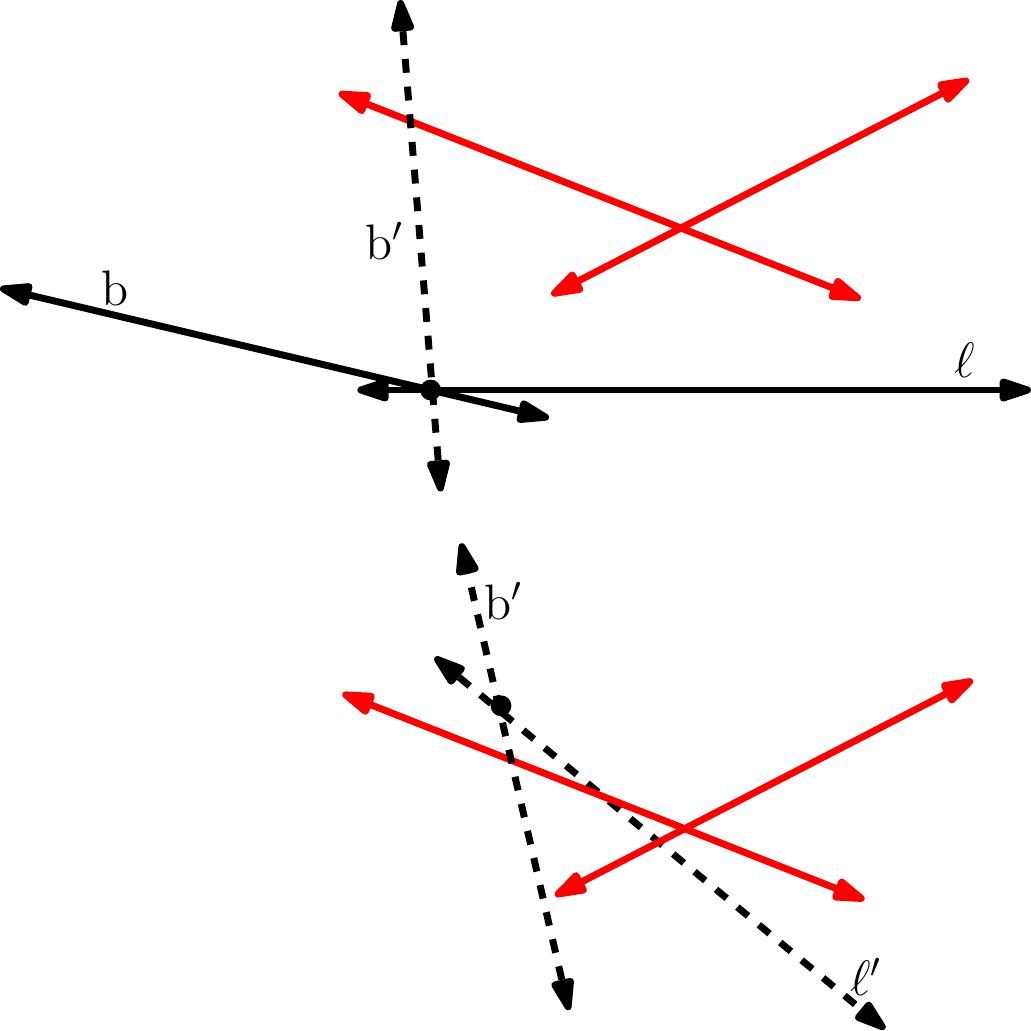} 
   \caption{The solid lines show the expected positions of $\ell$ and $b$ while the dashed lines show their possible true positions.  $\ell$ participates as a line in a tangent certificate with the red upper envelope and as a vertex in a diverging certificate.  There is no way for $\ell$ to be on the upper envelope without violating the tangent certificate it participates in.  In the top figure, the diverging certificate fails, but $\ell$ is not on the upper envelope, while in the bottom figure, both certificates fail and $\ell$ is on the resulting upper envelope.}
   \label{fig:failing_cert}
\end{figure}

  \looseness-1 Suppose that $\ell$ takes the role of $c$ in the tangent certificate (see Figure \ref{fig:failing_cert}): we will consider what other certificates $\ell$ could participate in.  Since $\ell$ already participates in one certificate as a line, it can not participate in any others as a line.  This is because $\ell$ can only have one type of slope relationship with the chain that is above it - either tangent, intersecting, or diverging.
That leaves participation as a vertex.  Depending on the number of lines participating in the merge at this level, $\ell$ may have zero, one, or two endpoints contributing to its upper envelope segment. 
At any existing endpoint, $\ell$ could participate as a vertex in a diverging certificate or an intersection certificate, but may participate in at most one certificate per vertex.  For either a failing diverging or intersection certificate to exclude $\ell$, the vertex of $\ell$ that is participating in the certificate must be above the line from the other (red) chain.  But if one of the endpoints of $\ell$ is above the other chain, then the tangent certificate is incorrect.  So for $\ell$ to be excluded, the tangent certificate it participates in must be incorrect.

A similar analysis when $\ell$ takes the role of $b$ in the diverging certificate argues that for $\ell$ to be excluded in that case, the diverging certificate it participates in as $b$ must be incorrect.  Finally, remember that if $\ell$ participates as a line in an intersection certificate it can not be excluded from the hull as it is already found to be on the hull (for this level).  We have considered all possible cases when $\ell$ participates as a line and in each of them there is a single certificate that must be wrong for $\ell$ to be excluded from the hull and no other certificate may exclude $\ell$ from the hull on its own.
\end{proof}

Now we can put these lemmas together for the proof of Theorem \ref{thm:convex_hull}.

\begin{proof}
By Lemmas \ref{lem:highest_level} and \ref{lem:one_cert} we know that each point can be excluded from the hull only by an incorrect certificate in its highest level of the merge tree and that each point has at most one certificate per level that can exclude it from the hull, so each point has at most one certificate in the whole merge tree that can exclude it.  If a convex hull is $\phi$-correct under the certificate error model, then each certificate has $1 - \phi$ probability of being incorrect and thus each point not on the hull has at most $1 - \phi$ probability of being excluded.  
By linearity of expectation, the expected number of points outside the reported hull is at most $(1 - \phi) n$.  So the convex hull is expected to be $\phi$-correct under the robust error model.
\end{proof}

\section{Conclusion}
In this paper, we introduced a new way of understanding and quantifying approximate geometric correctness via our certificate error model.  Accompanying this, we showed how this error model could be applied to probabilistic points of the type generated by machine learning models in the context of two problems - the 1D maximum and the convex hull problems.  We gave an $O(n \log n)$ time algorithm for the convex hull on probabilistic points, with approximate correctness guaranteed under the robust error model and under this certificate error model.  The strength of the certificate error model lies in its generalizability to any problem that can be stated in terms of its component Boolean properties.  These results could also be generalized to probabilistic points in a motion setting, since certificates are also the basis for kinetic data structure (KDS) results.

\bibliography{pkds}

\begin{thebibliography}{10}

\bibitem{Agarwal13ANN}
Pankaj~K. Agarwal, Boris Aronov, Sariel Har-Peled, Jeff~M. Phillips, Ke~Yi, and
  Wuzhou Zhang.
\newblock Nearest neighbor searching under uncertainty {II}.
\newblock In {\em PODS}, 2013.

\bibitem{Agarwal14ConvexHulls}
Pankaj~K. Agarwal, Sariel Har-Peled, Subhash Suri, Hakan Yıldız, and Wuzhou
  Zhang.
\newblock Convex hulls under uncertainty.
\newblock In {\em Proc. of the 22nd Annual European Symposium (ESA)}, pages
  37--48, 2014.

\bibitem{Arya00Approximate}
Sunil Arya and David~M. Mount.
\newblock Approximate range searching.
\newblock {\em Computational Geometry: Theory and Applications},
  17(3-4):135--152, 2000.

\bibitem{Atallah85Dynamic}
Mikhail~J. Atallah.
\newblock Some dynamic computational geometry problems.
\newblock In {\em Computers \& Mathematics with Applications}, volume 11(12),
  pages 1171--1181, 1985.

\bibitem{Basch99MobileData}
Julien Basch and Leonidas~J. Guibas.
\newblock Data structures for mobile data.
\newblock {\em J. of Algorithms}, 31(1):1--28, April 1999.

\bibitem{animalHomeRanges2010}
S.~Benhamou and D.~Corn\'{e}lis.
\newblock Incorporating movement behavior and barriers to improve kernel home
  range space use estimates.
\newblock {\em The Journal of Wildlife Management}, 74:1353--1360, 2010.

\bibitem{Fonseca10Approximate}
Guilherme~D. da~Fonseca and David~M Mount.
\newblock Approximate range searching: The absolute model.
\newblock {\em Computational Geometry: Theory and Applications},
  43(4):434--444, 2010.

\bibitem{cgbook}
Mark de~Berg, Otfried Cheong, Marc van Kreveld, and Mark Overmars.
\newblock {\em Computational Geometry: Algorithms and Applications}.
\newblock Springer Berlin Heidelberg, 2008.

\bibitem{devillers95incremental}
Olivier Devillers and Mordecai Golin.
\newblock Incremental algorithms for finding the convex hulls of circles and
  the lower envelopes of parabolas.
\newblock {\em Information Processing Letters}, 56(3):157 -- 164, 1995.

\bibitem{Guibas04Kinetic}
L.~Guibas.
\newblock Kinetic data structures.
\newblock In D.~Mehta and S.~Sahni, editors, {\em Handbook of Data Structures
  and Applications}, pages 23--1--23--18. Chapman and Hall/CRC, 2004.

\bibitem{Hightower01Location}
Jeffrey Hightower and Gaetano Borriello.
\newblock Location systems for ubiquitous computing.
\newblock {\em Computer}, 34(8):57--66, 2001.

\bibitem{Huber72Robust}
Peter~J. Huber.
\newblock The 1972 wald lecture robust statistics: A review.
\newblock {\em The Annals of Mathematical Statistics}, 43:1041--1067, 1972.

\bibitem{collisionDetection2001}
Pablo Jim\'{e}nez, Federico Thomas, and Carme Torras.
\newblock 3d collision detection: a survey.
\newblock {\em Computers \& Graphic}, 25(2):269--285, 2001.

\bibitem{Loeffler10Imprecise}
Maarten L{\"o}ffler and Marc van Kreveld.
\newblock Largest and smallest convex hulls for imprecise points.
\newblock {\em Algorithmica}, 56:235--269, 2010.

\bibitem{ORourke1998book}
Joseph O'Rourke.
\newblock {\em Computational Geometry in C}.
\newblock Cambridge University Press, 1998.

\bibitem{Rousseeuw05Robust}
Peter~J. Rousseeuw and Annick~M. Leroy.
\newblock {\em Robust Regression and Outlier Detection}.
\newblock Wiley \& Sons, NY, 2005.

\bibitem{imageProcBook2014}
Milan Sonka, Vaclav Hlavac, and Roger Boyle.
\newblock {\em Image processing, analysis, and machine vision.}
\newblock Cengage Learning, 2014.

\bibitem{augrealitySurvey2008}
Feng Zhou, Henry Been-Lirn Duh, and Mark Billinghurst.
\newblock Trends in augmented reality tracking, interaction and display: A
  review of ten years of {ISMAR}.
\newblock {\em Proceedings of the 7th IEEE/ACM International Symposium on Mixed
  and Augmented Reality}, pages 193--202, 2008.

\end{thebibliography}

\end{document}